\documentclass[a4paper,11pt,english]{article}
\usepackage[utf8]{inputenc}
\usepackage[margin=1.25in]{geometry}
\usepackage{babel,amsmath,amssymb,amsthm,enumitem}
\usepackage[sort]{natbib}
\usepackage[small,bf]{titlesec}
\titlelabel{\thetitle.\hspace{0.5em}}
\usepackage[colorlinks,allcolors=blue]{hyperref}

\renewcommand{\P}{\mathrm{P}}
\newcommand{\E}{\mathrm{E}}
\newcommand{\I}{\mathrm{I}}
\newcommand{\B}{\mathcal{B}}
\newcommand{\F}{\mathcal{F}}
\newcommand{\FF}{\mathbb{F}}
\newcommand{\R}{\mathbb{R}}
\newcommand{\midd}{\;\bigg|\;}
\renewcommand{\hat}{\widehat}
\renewcommand{\tilde}{\widetilde}
\renewcommand{\epsilon}{\varepsilon}
\newcommand{\RNP}{{\mathbb{R}^N_+}}
\renewcommand{\L}{\mathrm{L}}

\newtheorem{theorem}{Theorem}
\newtheorem{lemma}{Lemma}
\theoremstyle{definition}
\newtheorem{definition}{Definition}
\newtheorem{remark}{Remark}
\newtheorem{example}{Example}

\title{Relative growth optimal strategies in an asset market game}
\author{Yaroslav Drokin\thanks{National Research University Higher School of
Economics, Faculty of Mathematics. 6~Usacheva St., Moscow 119048, Russia.
E-mail: yadrokin@mail.ru.}
\and
Mikhail Zhitlukhin\thanks{Steklov Mathematical Institute of the Russian
Academy of Sciences. 8 Gubkina St., Moscow 119991, Russia. E-mail:
mikhailzh@mi-ras.ru. Corresponding author.
\newline\hspace*{\parindent} The research was supported by the Russian
Science Foundation, project 18-71-10097.}}

\date{23 July 2020}

\begin{document}
\maketitle

\begin{abstract}
We consider a game-theoretic model of a market where investors compete for
payoffs yielded by several assets. The main result consists in a proof of
the existence and uniqueness of a strategy, called relative growth
optimal, such that the logarithm of the share of its wealth in the total
wealth of the market is a submartingale for any strategies of the other
investors. It is also shown that this strategy is asymptotically optimal in
the sense that it achieves the maximal capital growth rate when compared to
competing strategies. Based on the results obtained, we study the asymptotic
structure of the market when all the investors use the relative growth
optimal strategy.



\medskip
\noindent
This paper was pubished in \textit{Annals of Finance} (2020).\\
\url{https://doi.org/10.1007/s10436-020-00360-6}
\end{abstract}

\section{Introduction}
Growth optimal strategies are a well-studied topic in mathematical finance.
However, the majority of models in the literature assume exogenously
specified returns of assets and consider models with a single investor. In
this setting growth optimal strategies arise as solutions of optimization
problems (see \cite{Breiman61,AlgoetCover88,Platen06,KaratzasKardaras07}).
In the present paper we study growth optimality of investment strategies
from a game-theoretic perspective and consider a model of a market (an \emph{asset
market game}), where several investors compete for random payoffs yielded by
several assets at discrete moments of time on the infinite time interval.
The payoffs are divided between the investors proportionally to shares of
assets they buy at prices determined endogenously by a short-run equilibrium
of supply and demand. As a result, the profit or loss of one investor
depends not only on the realized payoffs, but also on actions of the
competitors. In our model we assume that the assets are short-lived in the
sense that they are traded at time $t$, yield payoffs at $t+1$, and then the
cycle repeats. Thus, they can be viewed as some short-term investment
projects rather than, e.g., common stock.

The goal of the paper is to identify an investment strategy, called
relative growth optimal, such that the logarithm of the relative wealth of
an investor who uses it is a submartingale no matter what the strategies of
the other investors are (by relative wealth we mean the share of wealth of
one investor in the total wealth of the market). In conventional
(single-investor, i.e.\ non-game) market models, it is well-known that the
submartingale property implies various asymptotic optimality properties for
the growth rate of wealth (see, e.g.,
\cite{AlgoetCover88,KaratzasKardaras07}). Results of a similar nature turn
out to be true in our model as well, though their proofs use different
ideas. In particular, we show that the relative wealth of an investor who
uses the relative growth optimal strategy stays bounded away from zero
with probability one. We also show that if the representative strategy of
the other investors is asymptotically different, then such an investor will
dominate in the market -- the corresponding share of wealth will tend to
one. In addition, the relative growth optimal strategy maximizes the
growth rate of wealth and forms a symmetric Nash equilibrium in a game where
all investors maximize their expected relative wealth.

Our model extends the model proposed by \cite{AmirEvstigneev+13}, who also
studied  optimal strategies in an asset market game with
short-lived assets. The main difference between our model and their model is
that we assume the presence of a bank account (or a risk-free asset) with
exogenous interest rate. In the simplest form, if the interest rate is
zero, it just gives investors the possibility to put only a part of their
wealth in the assets and to keep a part of wealth in cash. In the model of
\cite{AmirEvstigneev+13}, it is assumed that the whole wealth is reinvested
in the assets in each time period, and that model can be obtained from ours
if one let the interest rate be $-1$, so that it is not reasonable to keep
money in the bank. The inclusion of a bank account in the model leads to a
more difficult construction of the optimal strategy. But, at he same time,
it also opens a series of new interesting questions regarding the asymptotic
behavior of the absolute wealth of investors that do not arise in the model
where the whole wealth is reinvested in assets.

Our paper can be reckoned among papers that study long-run performance of
investment strategies from the point of view of evolutionary dynamics, i.e.\
a market is considered as a population of various strategies which compete
for capital. This approach can be used to analyze forces that determine
long-run market dynamics through a process of natural selection of
investment strategies, see, e.g., the seminal paper by \cite{BlumeEasley92},
where a model with a discrete probability space was considered, and its
further development \citep{BlumeEasley06}. In later works this field is
called Evolutionary Finance; for recent literature reviews, see, e.g.,
\cite{EvstigneevHens+16,Holtfort19}.

The paper is organized as follows. In Section~\ref{sec-model}, we formulate
the model and introduce the notion of relative growth optimality of
investment strategies. In Section~\ref{sec-strategy}, we construct a
relative growth optimal strategy in an explicit form and show that it is
unique in a certain sense. Sections~\ref{sec-dominance}
and~\ref{sec-other-properties} study further optimality properties of this
strategy. Section~\ref{sec-wealth} is devoted to the analysis of the
asymptotics of the absolute wealth of investors when they use the relative
growth optimal strategy.

\section{The model}
\label{sec-model}
The market in the model consists of $M\ge 2$ investors, $N\ge1$ risky
assets, and a bank account (or cash). The assets yield payoffs which are
distributed between the investors at discrete moments of time $t=1,2,\ldots$
The investors choose, at every moment of time, proportions of their wealth
they invest in the assets and proportions they keep in the bank account. The
assets live for one period: they are traded at time $t$, yield payoffs at
$t+1$, and then the cycle repeats. Asset prices are determined endogenously
by a short-run equilibrium of supply and demand; in this model, without loss
of generality, we assume that each asset is in unit supply.

Let $(\Omega,\F,\P)$ be a probability space with a filtration
$\FF=(\F_t)_{t=0}^\infty$ on which all random variables will be defined.
Payoffs of asset $n=1,\ldots,N$ are specified by a random sequence $X_t^n\ge
0$, $t\ge 1$, which is $\FF$-adapted ($X_t^n$ is $\F_t$-measurable for all
$t\ge1$). It is assumed that $X_t^n$ are given exogenously, i.e. do not
depend on actions of the investors. Return on the bank account is specified
by an exogenous $\FF$-predictable sequence $\rho_t\ge 0$ (i.e.\ $\rho_t$ is
$\F_{t-1}$-measurable), such that $\rho_t-1$ is interpreted as spot interest
rate between moments of time $t-1$ and $t$. Note that the interest rate may
be negative. We will assume that
\begin{equation}
\rho_t + \sum_n X_t^n > 0\ \text{ a.s. for all}\ t\ge
1\label{positive-condition}
\end{equation}
(otherwise the model may degenerate as will become clear below).

The wealth of investor $m$ is described by an adapted random sequence
$Y_t^m\ge 0$. The quantity $Y_t^m$ is the budget that this investor can
allocate at time $t$ for investment in the assets and the bank account. We
assume that the initial budget $Y_0^m$ of each investor is non-random and
strictly positive. The wealth $Y_t^m$ at moments of time $t\ge 1$ depends on
investors' strategies and the asset payoffs.

A strategy of investor $m$ consists of vectors of investment proportions
$\lambda_t^m = (\lambda_t^{m,1},\allowbreak\ldots,\lambda_t^{m,N})$, $t\ge 1$,
according to which this investor allocates available budget towards purchase
of assets at time $t-1$. The proportion $1-\sum_n\lambda_t^{m,n}$ is
allocated in the bank account. We assume that short sales and borrowing from
the bank account are not allowed, so vectors
$\lambda_t=(\lambda_t^{m,n})$ belong to the set $\Delta =\{\lambda\in
\R_+^{MN} : \sum_n \lambda^{m,n} \le1\ \text{for each}\ m\}$.

At each moment of time, investment proportions are selected by the investors
simultaneously and independently, so the model represents a
simul\-ta\-neous-move $N$-person dynamic game, and the proportion vectors
$\lambda_t^m$ represent the investors' actions. These actions may depend on
the game history, and we define a strategy $\Lambda^m$ of investor
$m$ as a sequence of functions
\[
\Lambda_t^{m}(\omega, y_0,\lambda_1,\ldots,\lambda_{t-1})\colon \Omega\times
\R^M_+\times \Delta^{t-1}\to [0,1]^N,\qquad t\ge1,
\]
which are $\F_{t-1}\otimes\B(\R_+\times\Delta^{t-1})$-measurable and for all
$\omega,t,y_0,\lambda_1,\ldots,\lambda_{t-1}$ satisfy the condition
\[
\sum_n \Lambda_t^{m,n}(\omega, y_0,\lambda_1,\ldots,\lambda_{t-1}) \le 1.
\]
The argument $y_0\in \R_+^M$ corresponds to the vector of initial capital
$Y_0=(Y_0^1,\ldots,Y_0^M)$. The arguments $\lambda_s=(\lambda_s^{m,n})$,
$m=1,\ldots,M$, $n=1,\ldots,N$ are investment proportions selected by the
investors at past moments of time (for $t=1$, the function
$\Lambda_t^m(\omega,y_0)$ does not depend on $\lambda_s$). The value of the
function $\Lambda_t^m$ corresponds to the vector of investment proportions
$\lambda_t^{m}$. The measurability of $\Lambda_t^m$ in $\omega$ with respect
to $\F_{t-1}$ means that future payoffs are not known to the investors at
the moment when they decide upon their actions.

After selection of investment proportions by the investors at time $t-1$,
equilibrium asset prices $p_{t-1}^n$ are determined from the market clearing
condition that the aggregate demand of each asset is equal to the aggregate
supply, which is assumed to be 1. Since investor $m$ can purchase
$x_t^{m,n} = \lambda_t^{m,n} Y_{t-1}^m / p_{t-1}^n$ units of asset $n$,
the asset prices at time $t-1$ should be equal
\[
p_{t-1}^n = \sum_m \lambda_{t}^{m,n} Y_{t-1}^m.
\]
If $\sum_m\lambda_t^{m,n}=0$, i.e.\ no one invests in asset $n$, we put
$p_{t-1}^n=0$ and $x_t^{m,n}=0$ for all $m$.

Thus, investor $m$'s portfolio between moments of time $t-1$ and $t$
consists of $x_t^{m,n}$ units of asset $n$ and $c_t^{m} :=
(1-\sum_n\lambda_t^{m,n}) Y_{t-1}^m$ units of cash held in the bank account.
At moment of time $t$, the total payoff received by this investor from the
assets in her portfolio will be equal to $\sum_n x_t^{m,n} X_t^n$ and the
(gross) return on the bank account will be $\rho_tc_t^m$. Consequently, her
wealth is determined by the recursive relation
\begin{equation}
Y_{t}^m = \rho_t\Bigl(1-\sum_n \lambda_{t}^{m,n}\Bigr)Y_{t-1}^m + \sum_n
\frac{\lambda_t^{m,n} Y_{t-1}^m}{\sum_k
\lambda_{t}^{k,n} Y_{t-1}^k} X_t^n,\qquad
t\ge 1
\label{capital-equation-1}
\end{equation}
(with $0/0=0$ in the right-hand side). Here and in what follows, $Y_t^m$,
$Y_{t-1}^m$, $\rho_t$, $X_t^n$ are functions of
$\omega$ only, and by $\lambda_t^{m,n}$ we denote the \emph{realization} of
investor $m$'s strategy in this market, which is defined recursively as the
predictable sequence
\begin{equation}
\lambda_t^{m,n}(\omega) = \Lambda_t^{m,n}(\omega, Y_0, \lambda_1(\omega),
\ldots,\lambda_{t-1}(\omega)).\label{realization}
\end{equation}

Note that in our model investors' actions precede asset prices, so investors
first ``announce'' how much they allocate in each asset, and then the prices
are adjusted to clear the market. This modeling approach is analogous to
market games of Shapley--Shubik type. Although it is a simplification of a
real market, such an approach is economically reasonable (see
\cite{ShapleyShubik77} for details and justification).

We call \eqref{capital-equation-1} the \emph{wealth equation}, and it is the
principal equation in our model. Mainly, we will be interested not in the
absolute wealth $Y_t^m$, but in the relative wealth, i.e. the proportion of
wealth of one investor in the total wealth of all investors. The total
wealth is defined as
\[
W_t = \sum_m Y_t^m,
\]
and the relative wealth of investor $m$ is defined as
\[
r_t^m = \frac{Y_t^m}{W_t}
\]
(when $W_t=0$, we put $r_t^m=0$).

\begin{definition}
We call a strategy $\Lambda^m$ of investor $m$ \emph{relative growth
optimal} if for any vector of initial capital $Y_0$ (with $Y_0^k>0$ for all
$k$) and strategies $\Lambda^k$ of the other investors $k\neq m$,
\[
\ln r_t^m\text{ is a submartingale}.
\]
\end{definition}

Such a strategy is optimal in several aspects. First, observe that if a
strategy is relative growth optimal, then also $r_t^m$ is a submartingale by
Jensen's inequality. As a corollary, it is not hard to see that a strategy
profile in which every investor uses a relative growth optimal strategy is a
Nash equilibrium in the game where investors maximize $\E r_t^m$ at a fixed
moment of time $t$. This follows from the fact that if the strategies of
investors $k\neq m$ are relative growth optimal, then $r_t^m=1-\sum_{k\neq
m} r_t^k$ is a supermartingale. As we will show in
Section~\ref{sec-strategy}, a relative growth optimal strategy is unique, so
when every investor uses it, their relative wealth will remain constant.

Second, as will be shown in Section~\ref{sec-other-properties}, an
investor who uses a relative growth optimal strategy achieves the highest
growth rate of wealth compared to the other investors in the market. This
property is analogous to the growth optimality property in single-investor
market models (\cite{Breiman61,AlgoetCover88,Kelly56}; and others). However,
it is essential that we require the logarithm of relative wealth $\ln r_t^m$
to be a submartingale; the logarithm of wealth $\ln Y_t^m$ may be not a
submartingale (see Section~\ref{sec-wealth}).

Also, a relative growth optimal strategy belongs to the class of survival
strategies, which plays the central role in Evolutionary Finance (see
\cite{EvstigneevHens+16}) and is defined as follows.

\begin{definition}
A strategy $\Lambda^m$ of investor $m$ is called \emph{survival} if for any
strategies of the other investors
\[
\inf_{t\ge 0} r_t^m >0 \text{ a.s.}
\]
\end{definition}

An investor using a survival strategy cannot be driven out of the market
(even asymptotically) in the sense that her relative wealth always stays
bounded away from zero. In this definition, we use the terminology of
\cite{AmirEvstigneev+13}; note that, for example, \cite{BlumeEasley92} use
the term ``survival'' in a somewhat different meaning.  The fact that a
relative growth optimal strategy is survival readily follows from that
$\ln r_t^m$ is a non-positive submartingale, and hence it has a finite limit
$l = \lim_{t\to\infty} \ln r_t^m$ a.s.\ (see, e.g., Chapter~7 of
\cite{Shiryaev19en} for this and other results from the theory of
discrete-time martingales used in this paper). Therefore, $\lim_{t\to\infty}
r_t^m = e^l > 0$ a.s.

\begin{remark}[Relation to the Amir--Evstigneev--Schenk-Hopp\'e model]
Our model generalizes the model of Amir, Evstigneev, and Schenk-Hopp\'e
\citep{AmirEvstigneev+13}, where it was assumed that investors reinvest
their whole wealth in assets in each time period. That model can be obtained
as a particular case of our model by taking $\rho_t\equiv 0$, so that it is
never reasonable to keep money in the bank account.

However, despite similarity, construction of the optimal strategy in our
model turns out to be more difficult. In particular, the optimal strategy of
\cite{AmirEvstigneev+13} is \emph{basic} in the sense that its investment
proportions do not depend on past actions of investors, while, as we will
see in the next section, in our model they depend on current wealth of all
investors, which depends on their past actions. Moreover, our model opens
interesting questions about the asymptotic behavior of the total wealth of
investors, which do not arise when the whole wealth is reinvested in assets.
We consider these questions in Section~\ref{sec-wealth}.

Another model of this kind with short-lived assets and a risk-free asset was
considered by \cite{BelkovEvstigneev+17}, where the existence of a survival
strategy was established, and it was also shown that all basic survival
strategies are asymptotically equal. However, in that model asset payoffs
depend on ``money supply'' (amount of capital not invested in assets) in a
special way, which allows to reduce that model to the one of
\cite{AmirEvstigneev+13}.
\end{remark}

\section{Existence and uniqueness of a relative growth optimal strategy}
\label{sec-strategy}
Let us first introduce auxiliary notation and definitions.

We will use the following convenient notation for vectors. If $x,y\in \R^N$,
we will denote their scalar product by $xy=\sum_n x^ny^n$, the $\L^1$-norm
by $|x|=\sum_n |x^n|$, and the $\L^2$-norm by $\|x\| = \sqrt{xx}$. If
$f\colon \R\to \R$ is a function, then $f(x)$ denotes the vector
$(f(x^1),\ldots,f(x^N))$. By $a\vee b$ we will denote the maximum of
variables $a,b$, and by $a\wedge b$ the minimum.

The realization of investor $m$'s strategy, defined in \eqref{realization},
will be denoted by $\lambda^m_t(\omega)$, and when it is necessary to
emphasize that it depends on the initial capital and strategies of the other
investors we will use the notation $\lambda_t^m(\omega; Y_0, L)$, where $L =
(\Lambda^1,\ldots,\Lambda^M)$ stands for a strategy profile; the argument
$\omega$ will be often omitted for brevity.

Let us introduce the notion of equality of strategies that will be used to
state that the relative growth optimal strategy is unique. Suppose
$\tau(\omega;Y_0, L)$ denotes a family of random variables, i.e. for any
fixed vector of initial capital $Y_0$ and a strategy profile
$L=(\Lambda^1,\ldots,\Lambda^M)$, the function
$\omega\mapsto\tau(\omega;Y_0,L)$ is $\F$-measurable (and may assume the
value $+\infty$).

\begin{definition}
We say that two strategies $\Lambda^m$ and $\tilde \Lambda^m$ of investor
$m$ are \emph{equal in realization until $\tau$} if for any vector of
initial capital $Y_0$ and any strategies of the other investors $\Lambda^k$,
$k\neq m$, we have the equality of realizations (a.s. for all $t\ge1$)
\[
(\lambda_t^m(\omega; Y_0, L) - \lambda_t^m(\omega; Y_0, \tilde
L))\I(t\le \tau(\omega; Y_0,L)\wedge \tau(\omega; Y_0, \tilde L))=0,
\]
where $L = (\Lambda^{1},\ldots,\Lambda^m,\ldots,\Lambda^{M})$, $\tilde L =
(\Lambda^{1},\ldots,\tilde\Lambda^m,\ldots,\Lambda^{M})$ are the strategy
profiles which differ only in the strategy of investor $m$.
\end{definition}

For example, in Theorem~\ref{th-1} below, we will consider equality until
$\tau = \inf\{t\ge 0 : r_t^m = 1\}$ -- the first moment of time when the
relative wealth of an investor reaches 1 (with $\tau(\omega)=\infty$ if
$r_t^m(\omega) <1$ for all $t$), which, considered as a function
$\omega\mapsto \tau(\omega; Y_0,L)$, is a stopping time for any fixed $Y_0$
and $L$.

The reason why we need to work with equality until $\tau$ is
that when the relative wealth of an investor becomes 1 (and the wealth of
the competitors becomes zero), she may choose any investment proportions and
her relative wealth will always remain 1 (provided that she does not invest
in a ``bad way'' losing all her wealth). Hence, it is not possible to speak
about uniqueness after this moment.

Now we can proceed to the construction of the relative growth optimal
strategy. Let $K_t(\omega, A)\colon \Omega\times \B(\R^N_+)\to[0,1]$ denote
the regular conditional distribution of the payoff vector
$X_t=(X_t^1,\ldots,X_t^N)$ with respect to $\F_{t-1}$, so that for each $t$
and fixed $\omega$ the function $A\mapsto K_t(\omega,A)$ is a probability
measure on the Borel $\sigma$-algebra $\B(\R_+^N)$, and for fixed $A\in
\B(\R^N_+)$ the function $\omega \mapsto K_t (\omega,A)$ is a version of the
conditional probability $\P(X_t\in A\mid \F_{t-1})$.

Define the sequence of sets $\Gamma_t \in \F_{t-1}\otimes\B(\R_+)$,
\[
\Gamma_t = \biggl\{(\omega,c) : \int_{\R_+^N} \frac{c\rho_t(\omega)}{|x|}
K_t(\omega,dx) > 1\biggr\}, \qquad
t\ge 1,
\]
with the following convention: $c\rho_t(\omega)/|x| = 0$ if
$c\rho_t(\omega)=|x|=0$ and $c\rho_t(\omega)/|x| = +\infty$ if
$c\rho_t(\omega)>0$ but $|x|=0$.

The following lemma will play an auxiliary role in construction of the
relative growth optimal strategy.

\begin{lemma}
\label{lem-1}
For all $t\ge 1$ and $(\omega,c)\in \Gamma_t$, there exists a unique
solution $z\in (0,c]$ of the equation
\begin{equation}
\int_{\R_+^N} \frac{c\rho_t(\omega)}{z\rho_t(\omega)+|x|} K_t(\omega, dx) =
1. \label{zeta}
\end{equation}
The function $\zeta_t(\omega,c)$ defined to be equal to this solution on
$\Gamma_t$ and equal to zero outside $\Gamma_t$ is $\F_{t-1}\otimes
\B(\R_+)$-measurable.
\end{lemma}
\begin{proof}
The existence and uniqueness of the solution for each $(\omega,c)\in
\Gamma_t$ is straightforward: the left-hand side of \eqref{zeta} is a
continuous and strictly decreasing function in $z\in(0,c]$ which assumes a
value greater than 1 for $z=0$ and a value not greater than 1 for $z=c$.

To prove the measurability, consider the function
$f\colon\Omega\times\R_+\times\R_+\to \R$,
\[
f(\omega,c,z) = \biggl(2\wedge\int_{\R_+^N}
\frac{c\rho_t(\omega)}{z\rho_t(\omega)+|x|} K_t(\omega, dx) -
1\biggr) \I((\omega,c)\in \Gamma_t).
\]
Observe that $f$ is a Carath\'eodory function, i.e.
$\F_{t-1}\otimes\B(\R_+)$-measurable in $(\omega,c)$ and continuous in $z$.
Then by Filippov's implicit function theorem (see, e.g., Theorem~18.17 in  
\cite{AliprantisBorder06}), the set-valued function
\[
\phi(\omega,c) = \{z\in[0,c]  : f(\omega,c,z)=0\}
\]
admits a measurable selector. Since $\phi$ on $\Gamma_t$ is single-valued
(we have $\phi(\omega,c) = \{\zeta_t(\omega,c)\}$), this implies that
$\zeta_t$ is $\F_{t-1}\otimes\B(\R_+)$-measurable.
\end{proof}

In what follows, we will use the notation $\chi_t
=(y_0,\lambda_1,\ldots,\lambda_t)\in \R_+^M\times \Delta^t$ for history of
the market until time $t$. Denote by
$C_t(\omega,\chi_t)=W_t(\omega)=|Y_t(\omega)|$ the total wealth of all
investors at time $t$, where $Y_t$ is defined recursively by relation
\eqref{capital-equation-1} with the given initial wealth $Y_0=y_0$ and
investment proportions $\lambda_s^{m,n}$, which form the history $\chi_t$.

\begin{theorem}
\label{th-1}
The strategy $\hat \Lambda$ with investment proportions defined by the
relation 
\begin{equation}
\hat \Lambda_t^n(\omega,\chi_{t-1}) = \int_{\R_+^N}
\frac{x^n}{\zeta_t(\omega,C_{t-1}(\omega,\chi_{t-1}))\rho_t(\omega) + |x|}
K_t(\omega,dx)\label{optimal}
\end{equation}
is relative growth optimal $($$0/0=0$ in \eqref{optimal}$)$.

Moreover, $\hat\Lambda$ is the unique relative growth optimal strategy in
the sense that if $\Lambda$ is another strategy for investor $m$ such that
its relative wealth $r_t^m$ is a submartingale for any initial capital and
strategies of the other investors, then $\hat \Lambda$ and $\Lambda$ are
equal in realization until the time $\tau = \inf\{t\ge 0 : r_t^m = 1\}$.
\end{theorem}

It is easy to see that the proportion of wealth that $\hat\Lambda$ keeps in
the bank account is $1-|\hat\Lambda(\chi_{t-1})| =
\zeta_t(C_{t-1}(\chi_{t-1}))/C_{t-1}(\chi_{t-1})$. In particular, if
$\rho_t=0$ for all $t$, then $\Gamma_t=\emptyset$ and $\zeta_t = 0$. In this
case we obtain the same strategy that was found by \cite{AmirEvstigneev+13}
-- it divides the available budget between the assets proportionally to
their expected \emph{relative payoffs} $\int_\RNP x^n / |x| K_t(dx) = \E(X_t^n/|X_t|
\mid \F_{t-1})$. When $\zeta_t\neq 0$, the strategy $\hat\Lambda$ still
divides the budget between the assets proportionally to their payoffs but
the proportions are adjusted for the amount of capital kept in the bank
account.

\bigskip Before we proceed to the proof of Theorem~\ref{th-1}, let us state
one auxiliary inequality that we will use (it generalizes Gibbs'
inequality).
\begin{lemma}
\label{lemma-logsum}
Suppose $\alpha,\beta\in\RNP$ are two vectors such that
$|\alpha|,|\beta|\le 1$ and for each $n$ it
holds that if $\beta^n=0$, then also $\alpha^n=0$. Then 
\begin{equation}
\alpha(\ln\alpha - \ln \beta) \ge \frac{\|\alpha-\beta\|^2}{4} + |\alpha|-|\beta|,\label{logsum-1}
\end{equation}
where we put $\alpha^n(\ln \alpha^n-\ln \beta^n) = 0$ if $\alpha^n=0$.
\end{lemma}
\begin{proof}
We follow the lines of the proof of Lemma~2 in~\cite{AmirEvstigneev+13},
which establishes the above inequality in the case $|\alpha|=|\beta|=1$.
Using that $\ln x \le 2(\sqrt x -1)$ for any $x>0$, we obtain
\[
\begin{split}
\alpha(\ln\alpha-\ln\beta) &= -\sum_{n\,:\,\alpha^n\neq 0}\alpha^n\ln(\beta^n/\alpha^n) \ge
2\sum_n (\alpha^n-\sqrt{\alpha^n\beta^n}) \\&= \sum_n(\sqrt{\alpha^n} -
\sqrt{\beta^n})^2 + |\alpha| - |\beta|.
\end{split}
\]
Then we can use the inequality $(\sqrt x - \sqrt y)^2 \ge (x-y)^2/4$, which
is true for any $x,y\in[0,1]$, and obtain \eqref{logsum-1}. 
\end{proof}

\begin{proof}[Proof of Theorem~\ref{th-1}]
Without loss of generality, assume that the strategy $\hat\Lambda$ is used
by investor 1. Let $\lambda_t$ denote the realization of this strategy, and
$\tilde \lambda_t$ the realization of the representative strategy of the
other investors, which we define as the following weighted sum of the
realizations of their strategies:
\begin{equation}
\tilde \lambda_t^n = \sum_{m\ge 2} \frac{r_{t-1}^m}{1-r_{t-1}^1}\lambda_t^{m,n},\label{repr-realization}
\end{equation}
where $\tilde \lambda_t = 0$ if $r_{t-1}^1=1$. By $Y_t$ we will denote the
wealth of investor 1, by $\tilde Y_t:=\sum_{m\ge2} Y_t^m$ the total wealth
of the other investors, and by $r_t = Y_t/(Y_t+\tilde Y_t)$ the relative
wealth of investor 1. Then $Y_t$ satisfies the following relation, which
follows from \eqref{capital-equation-1}:
\begin{equation}
Y_t = \rho_t(1-|\lambda_t|)Y_{t-1} + \sum_n \frac{\lambda_t^{n}
Y_{t-1}}{\lambda_t^{n}Y_{t-1} + \tilde \lambda_t^{n} \tilde Y_{t-1}}
X_t^n.\label{eq-proof-1}
\end{equation}
Observe that from the definition of $\hat\Lambda$ and
condition~\eqref{positive-condition}, it follows that for each $t$ and
almost all $\omega$ we have
\begin{equation}
K_t(\omega, \{x: \rho_t(\omega)(1-|\lambda_t(\omega)|)+x
\lambda_t(\omega)=0\})=0\label{proof-K}.
\end{equation}
In particular, this implies that $Y_t>0$ a.s. for all $t$.

Introduce the predictable sequence of random vectors $F_t$ with values in
$\R_+^N$ which have the components
\[
F_t^n = \frac{\lambda_t^n}{r_{t-1}\lambda_t^n + (1-r_{t-1}) \tilde
\lambda_t^n},
\]
where $0/0=0$. Denoting the total wealth of the
investors by $W_t=Y_t+\tilde Y_t$, the equation \eqref{eq-proof-1} can be
rewritten as
\begin{equation}
\label{proof-Y}
Y_t = \biggl(\rho_t(1- |\lambda_t|) + \frac{F_t X_t}{W_{t-1}}\biggr) Y_{t-1}.
\end{equation}
A similar equation is true for $\tilde Y_t$, namely, $\tilde Y_t =
(\rho_t(1- |\tilde \lambda_t|) + \frac{\tilde F_t X_t}{W_{t-1}} )\tilde
Y_{t-1}$, where $\tilde F_t^n = \tilde \lambda_t^n/(r_{t-1}\lambda_t^n +
(1-r_{t-1}) \tilde \lambda_t^n)$. Using this, we obtain
\begin{equation}
\label{proof-W}
W_t = \biggl(\rho_t(1 - r_{t-1}|\lambda_t| - (1-r_{t-1}) |\tilde\lambda_t|)
+ \frac{|X_t|}{W_{t-1}}\biggr)W_{t-1}.
\end{equation}
In this equation we used the equality $(r_{t-1}F_t^n+(1-r_{r-1})\tilde F_t^n)X_t^n =
X_t^n$: on the set $\{\lambda_t^n>0\}$ this is clear from the definition of
$F_t^n$ and $\tilde F_t^n$, while on the set $\{\lambda_t^n=0\}$ we have
$X_t^n=0$ a.s., which follows from the construction of $\hat \Lambda$.

Let $\zeta_t(\omega)$ denote the predictable sequence
$\zeta_t(\omega,W_{t-1}(\omega))$. As follows from the definition of
$\hat\Lambda$, we have $|\lambda_t| = 1- \zeta_t/W_{t-1}$. Let $\tilde
\zeta_t = (1-|\tilde \lambda_t|) W_{t-1}$.
\label{r-ratio}
Then dividing \eqref{proof-Y} by \eqref{proof-W} we find that $\ln r_t - \ln
r_{t-1} = f_t(X_t)$, where $f_t = f_t(\omega,x)$ is the
$\F_{t-1}\otimes\B(\R_+^N)$-measurable function
\[
f_t(x) = \ln \biggl(\frac{\zeta_t\rho_t + F_tx}{r_{t-1}\zeta_t\rho_t +
(1-r_{t-1})\tilde\zeta_t\rho_t + |x|}\biggr)
\]
(for brevity, the argument $\omega$ will be omitted). Note that
\eqref{proof-K} implies $\zeta_t\rho_t + F_tX_t > 0$ a.s., hence we can
define the value of $f_t(\omega,x)$ for $x$ such that
$\zeta_t(\omega)\rho_t(\omega) + F_t(\omega)x=0$ in an arbitrary way. It
will be convenient to put $f_t(\omega,x) = 0$ for such $x$.

To show that $\ln r_t$ is a submartingale, it will be enough to show that
$\int_\RNP f_t(x) K_t(dx) \ge 0$, i.e. $\E( f_t(X_t) \mid \F_{t-1}) \ge 0$.
Indeed, then $\ln r_t$ will be a generalized submartingale\footnote{Recall
that a sequence $S_t$ is called a generalized submartingale if $\E
|S_0|<\infty$ and $\E(S_t\mid\F_{t-1})\ge S_{t-1}$ for all $t\ge 1$ (but not
necessarily $\E|S_t|<\infty$). It is easy to show that if $S_t\le C_t$ for
all $t$ with some integrable random variables $C_t$, then $S_t$ is
integrable, and hence a usual submartingale.}, but
since it is bounded from above (by 0), this will also imply that $\ln r_t$
is a usual submartingale (see Chapter~7.1 in \cite{Shiryaev19en}).

Suppose for some $t,\omega$, a vector $x$ is such that
$\zeta_t(\omega)\rho_t(\omega) + F_t(\omega)x>0$, and, for each $n$, the
equality $F_t^n(\omega) = 0$ implies $x^n=0$. Then we have the bound
\[
\begin{split}
f_t(x) &= \ln\biggl(\frac{\zeta_t\rho_t + F_tx}{\zeta_t\rho_t + |x|}\biggr) +
\ln\biggl(\frac{\zeta_t\rho_t + |x|}{r_{t-1}\zeta_t\rho_t + (1-r_{r-1})\tilde\zeta_t\rho_t +
|x|}\biggr) \\ &\ge \frac{x\ln F_t}{\zeta_t\rho_t + |x|} +
\frac{(1-r_{t-1})(\zeta_t-\tilde\zeta_t)\rho_t}{\zeta_t\rho_t + |x|} := g_t(x) + h_t(x),
\end{split}
\]
where we put $x^n\ln F_t^n = 0$ if $F_t^n=0$. Here, for the first term in
the second line we used the concavity of the logarithm, and for the second
term the inequality $\ln a \ge 1-a^{-1}$.

For each $t$, we have $K_t(\{x: \zeta_t\rho_t+x F_t=0\})=0$ a.s.
by~\eqref{proof-K}, and also $K_t(\{x^n=0\})=1$ a.s. on the set
$\{F_t^n=0\}$ by~the definition of $\hat\Lambda$. Hence
\begin{equation}
\label{proof-int-f}
\int_\RNP f_t(x) K_t(dx) \ge\int_\RNP g_t(x) K_t(dx) + \int_\RNP h_t(x)
K_t(dx) := I_t^g + I_t^h.
\end{equation}
For the  integral $I_t^g$, using Lemma~\ref{lemma-logsum}, we find
\begin{equation}
\label{proof-g}
\begin{split}
I_t^g &= \lambda_t \ln F_t = \lambda_t( \ln \lambda_t -
\ln(r_{t-1}\lambda_t + (1-r_{t-1})\tilde\lambda_t)) \\&\ge
\frac14 (1-r_{t-1})^2\|\lambda_t-\tilde\lambda_t\|^2 + (1-r_{t-1})(|\lambda_t| - |\tilde \lambda_t|).
\end{split}
\end{equation}
For the integral $I_t^h$, on the set $\{\omega: (\omega,W_{t-1}(\omega)) \in
\Gamma_t\}$ we can use the equality $\int_\RNP \rho_t W_{t-1}/(\zeta_t\rho_t
+ |x|) K_t(dx) = 1$, and on its complement the equality $\zeta_t=0$ and
inequality $\int_\RNP \rho_t W_{t-1}/|x| K_t(dx) \le 1$, which result in
\begin{equation}
\label{proof-h}
I_t^h \ge
{(1-r_{t-1})(\zeta_t-\tilde\zeta_t)}/{W_{t-1}} = (1-r_{t-1})
(|\tilde\lambda_t| - |\lambda_t|).
\end{equation}
Relations \eqref{proof-int-f}--\eqref{proof-h} imply $\int_\RNP f_t(x)
K_t(dx) \ge 0$, so $\ln r_t$ is a submartingale, which proves that
$\hat\lambda$ is relative growth optimal.

In order to prove the statement about uniqueness, suppose there exists
another strategy $\Lambda'$ (without loss of generality, assume this is a
strategy of investor 1) with relative wealth $r_t^1$ being a submartingale
for any initial capital and strategies of the other investors, and there
exist strategies $\Lambda^2,\ldots,\Lambda^M$ and a vector of initial
capital $Y_0$ such that the realizations of $\hat\Lambda$ and $\Lambda'$ are
different for the strategy profiles $\hat
L=(\hat\Lambda,\Lambda^2,\ldots,\Lambda^M)$ and
$L'=(\Lambda',\Lambda^2,\ldots,\Lambda^M)$, i.e. $\P(\hat\lambda_t \neq
\lambda'_t)>0$ for some $t$, where $\hat\lambda_t(\omega)$,
$\lambda_t'(\omega)$ are the realizations of the strategy of the first
investor in the markets $\hat L$ and $L'$, respectively.

Consider the predictable stopping time $\sigma = \inf \{t \ge 1 :
\hat\lambda_t \neq \lambda'_t\}$ and define the new strategies $\tilde
\Lambda^m$, $m\ge 2$, by
\[
\tilde\Lambda_t^{m} (\omega, \chi_{t-1}) =
(\Lambda')_t^m(\omega,\chi_{t-1}) \I(t<\sigma(\omega)) +\hat
\Lambda_t^m(\omega, \chi_{t-1})\I(t\ge\sigma(\omega)),
\]
where $\chi_{t-1} = (y_0,\lambda_1,\ldots,\lambda_{t-1})$ denotes history of
the market until $t-1$. Observe that since $\{\sigma\le t\} \in \F_{t-1}$,
the strategies are well-defined (i.e. $\F_{t-1}\otimes \B(\R^M_+\times
\Delta^{t-1})$-measurable).

Let $r_t(\omega)$ denote the realization $r_t^1(\omega,Y_0,\tilde L)$ of the
relative wealth of investor 1 when the investors use the strategy profile
$\tilde L = (\Lambda', \tilde\Lambda^2,\ldots,\tilde\Lambda^M)$. Then, on
one hand, on the set $\{\sigma<\infty\}$ we have $\E(r_{\sigma}\mid
\F_{\sigma-1}) \ge r_{\sigma-1}$ by the choice of $\Lambda'$, and hence
$\E(\ln(1-r_{\sigma})\mid \F_{\sigma-1}) \le \ln(1-r_{\sigma-1})$ by
Jensen's inequality. On the other hand, $\E(\ln(1 - r_{\sigma})\mid
\F_{\sigma-1}) \ge \ln(1-r_{\sigma-1})$ since investors $m\ge 2$ use the
strategy $\hat\Lambda$ after $\sigma$. Hence, on the set $\{\sigma<\infty,\;
r_{\sigma-1}<1\}$ we have
\begin{equation}
0 = \E\biggl( \ln \frac{1 - r_{\sigma}}{1 - r_{\sigma-1}}\;\bigg|\; \F_{\sigma-1}\biggr) \ge
 \frac14 r_{\sigma-1}^2\|\lambda_{\sigma}-\tilde\lambda_{\sigma}\|^2,\label{proof-uniq}
\end{equation}
where the inequality can be obtained from \eqref{proof-g} and
\eqref{proof-h}. In this formula, $\lambda_\sigma$ and
$\tilde\lambda_\sigma$ are the realizations of investor 1's strategy
$\Lambda'$ and the representative strategy of investors $m\ge 2$ at time
$\sigma$ in the market with the strategy profile $\tilde L$ . It is not hard
to see that on the set $\{\sigma<\infty\}$ they are equal, respectively, to
$\lambda'_\sigma$ and $\hat\lambda_\sigma$. Consequently, \eqref{proof-uniq}
and the choice of $\sigma$ imply that if the set $\{\sigma<\infty,
r_{\sigma-1}<1\}$ had positive probability, then $r_{\sigma-1}=0$ a.s. on
it. But this is impossible since up to $\sigma-1$ the realization of the
strategy $\Lambda'$ coincides with the realization of $\hat\Lambda$, and
hence its relative wealth stays positive. Thus, $\P(\sigma<\infty,
r_{\sigma-1}<1) = 0$, which proves the claimed uniqueness.
\end{proof}

\section{The dominance property of a relative growth optimal strategy}
\label{sec-dominance}
The next short result shows that the relative wealth of the relative
growth optimal strategy tends to 1 on the set of outcomes $\omega$ for which
the realization $\tilde\lambda$ of the representative strategy of the other
investors is asymptotically different from the realization of $\hat\Lambda$
in a certain sense ($\tilde \lambda$ is defined in \eqref{repr-realization}
above). This result can be viewed as asymptotic uniqueness of a survival
strategy.

\begin{theorem}
\label{th-2}
Suppose investor 1 uses the relative growth optimal strategy. Let
$\hat\lambda_t$ denote its realization, $\tilde \lambda_t$ denote the
realization of the representative strategy of the other investors, and
$\Omega'\in \F$ be the set
\[
\Omega' = \Bigl\{\omega : \sum_{t\ge 1}
\|\hat\lambda_t(\omega)-\tilde\lambda_t(\omega)\|^2 = \infty\Bigr\}.
\]
Then $r_t^1(\omega) \to 1$ a.s.\ on $\Omega'$. In particular,
$\|\hat\lambda_t(\omega)-\tilde\lambda_t(\omega)\| \to 0$ a.s.\ on the set
$\{\omega: \lim_{t\to\infty}r_t^1(\omega) < 1\}$.
\end{theorem}

\begin{proof}
We will use the same notation as in the proof of Theorem~\ref{th-1}. Since
$\ln r_t$ is a non-positive submartingale, it converges a.s.\ and its
compensator $C_t$ converges a.s.\ as well. From the proof of
Theorem~\ref{th-1}, it follows that
\[
C_{t} - C_{t-1} = \int_\RNP f_t(x) K_t(dx) \ge
\frac14 (1-r_{t-1})^2\|\hat\lambda_t-\tilde\lambda_t\|^2.
\]
Then on the set $\Omega'$ we necessarily have $r_{t} \to 1$ a.s., since
otherwise $C_t$ would diverge.
\end{proof}

\section{Maximization of the growth rate of wealth}
\label{sec-other-properties}
Recall that in a single-investor market model a \emph{num\'eraire portfolio} is a
strategy such that the ratio of the wealth of any other strategy to the
wealth of this strategy is a supermartingale. The term ``num\'eraire
portfolio'' was introduced by \cite{Long90}; often it is also called a
growth optimal strategy, or a benchmark portfolio (see
\cite{HakanssonZiemba95,KaratzasKardaras07,Platen06}). It is well-known that
num\'eraire portfolios have a number of optimality properties: they maximize
the asymptotic growth rate of wealth, maximize the expected logarithmic
utility, minimize the time to reach a given wealth level, etc.\ (see, e.g.,
\cite{Breiman61,AlgoetCover88} for results in discrete time, and
\cite{KaratzasKardaras07} for results in a general semimartingale model,
including a connection with the arbitrage theory). In this section we will
show that the relative growth optimal strategy in our model has similar
properties.

By the asymptotic growth rate of the wealth $Y_t$  we call
$\limsup_{t\to\infty} \frac1t\ln Y_t$ (see, e.g., Chapter~3.10 in
\cite{KaratzasShreve98}), and the $t$-step growth rate at time $s$ can be
defined as $\frac1t \E (\ln \frac{Y_{s+t}}{Y_s} \mid \F_{s})$. These notions
have especially clear interpretation in a single-investor model with i.i.d.\ 
asset returns: then the log-returns of a growth optimal strategy are i.i.d.\ 
as well, and the asymptotic growth rate and the $t$-step growth rate are
equal and non-random.

\begin{theorem}
\label{th-3}
Suppose investor 1 uses the relative growth optimal strategy and the other
investors use arbitrary strategies. Then for any vector of initial capital
$Y_0$ the following claims are true.

\noindent
1) Investor 1 maximizes the asymptotic growth rate of wealth: for any $m$
\begin{equation}
\limsup_{t\to\infty} \frac1t{\ln Y_t^1} \ge \limsup_{t\to\infty} \frac1t{\ln
Y_t^m}\ \text{a.s.}\label{opt-1}
\end{equation}
2) Suppose there are only two investors ($M=2$) and $\E|X_t|<\infty$ for all
$t$. Then investor 1 maximizes the $t$-step growth rate of wealth: for any
$t,s\ge 0$
\begin{equation}
\E\biggl( \ln\frac{Y_{s+t}^1}{Y_s^1}\midd \F_s\biggr) \ge \E\biggl(
\ln\frac{Y_{s+t}^2}{Y_s^2}\midd \F_s\biggr)\ \text{a.s.}\label{opt-2}
\end{equation}
(In inequalities~\eqref{opt-1},~\eqref{opt-2}, the both sides may assume the
values $\pm\infty$; in \eqref{opt-2} we put $\ln0/0 = -\infty$.)

\end{theorem}
\begin{proof}
1) As was noted above, if investor $1$ uses the relative growth optimal
strategy, then $\inf_t r_t^1>0$ a.s., and hence $\sup_t W_t/Y_t^1 < \infty$.
Therefore, $\sup_t Y_t^m/Y_t^1 < \infty$ for any $m$. This implies that for
any sequence $T_t$ such that $\lim_{t\to\infty} T_t = +\infty$ (in
particular, for $T_t=t$) we have the inequality
\[
\limsup_{t\to \infty} \frac{1}{T_t}\ln \frac{Y_t^m}{Y_t^1} \le 0\ \text{ a.s.}
\]
From here, one can obtain~\eqref{opt-1}.

2) From the condition $\E |X_t| < \infty$, it follows that $\E \ln W_t <
+\infty$ for all~$t$. On the set $\{\omega:\E(\ln W_{s+t} \mid \F_s)(\omega)
= -\infty\}$ we have $\E(\ln Y_{s+t}^1 \mid \F_s) = \E(\ln Y_{s+t}^2 \mid
\F_s)= -\infty$, so inequality \eqref{opt-2} holds on this set. On the set
$\{\omega:\E(\ln W_{s+t} \mid \F_s)(\omega) > -\infty\}$, use the
submartingale property of $\ln r_t^1$, which implies
\[
\E \biggl(\ln \frac{Y_{s+t}^1}{Y_s^1}  \midd \F_s\biggr) \ge 
\E\biggl(\ln\frac{W_{s+t}}{W_s} \midd \F_s \biggr).
\]
Since $r_t^2=1-r_t^1$ is a supermartingale, $\ln r_t^2$ is a generalized
supermartingale, so, in a similar way,
\[
\E \biggl(\ln \frac{Y_{s+t}^2}{Y_s^2}  \midd \F_s\biggr) \le 
\E\biggl(\ln\frac{W_{s+t}}{W_s} \midd \F_s \biggr),
\]
which proves \eqref{opt-2}. 
\end{proof}

\begin{remark}
1. It is clear from the proof that \eqref{opt-1} holds if investor 1 uses
any survival strategy.

2. Note that the second claim of Theorem~\ref{th-3} generally does not hold
in the case $M\ge 3$. For example, it can happen that investor 1 uses the
strategy $\hat\Lambda$, investor 2 acts in an unoptimal way, and investor 3
manages to find a strategy which is better than $\hat \Lambda$.

A simple example can be constructed even for a non-random market. Let $M=3$,
$N=1$, $Y_0^m=1$ for $m=1,2,3$, and $\rho_1=1$, $X_1=1$. Then we have
$\hat\lambda_1 = 1/3$. However, if $\lambda^1_1=\hat\lambda_1$ and
$\lambda^2_1 = 1$, then the strategy $\lambda^3_1 = 0$ turns out to be
better than $\lambda_1^1$ after one step: $Y_1^1=11/12$, but $Y_1^3=1$.
\end{remark}

\section{Growth of wealth when investors use the relative
growth optimal strategy}
\label{sec-wealth}
It is interesting to observe that using the relative growth optimal
strategy does not necessarily imply ``favorable'' asymptotics of the
absolute wealth.

We begin with an example which shows that the wealth of an investor who uses
the strategy $\hat\Lambda$ may vanish asymptotically because the other
investors use ``bad'' strategies such that the total wealth $W_t$ vanishes.
At the same time, there is a strategy the wealth of which does not vanish,
but it is not relative growth optimal. In the second part of this section,
we investigate the case when all the investors use the relative growth
optimal strategy; in that case their wealth will normally grow.

\begin{example}
We consider a non-random model with two investors and one asset. Suppose the
investors have the initial capital $Y^1_0=Y^2_0 = 1$ and use the strategies
that invest the proportions $\lambda_{t+1}^1 = \frac12$ and $\lambda_{t+1}^2
= \frac12+\frac1{2t}$ in the asset in each time period. Suppose $\rho_t=1$
for all $t$, and the (non-random) payoff sequence $X_t$ is defined by
\[
X_{t+1} = \frac{Y_t^1+Y_t^2}{2}.
\]
This equation together with the wealth equation \eqref{capital-equation-1}
uniquely define the sequences $Y_t^1$, $Y_t^2$, and $W_t = Y_t^1 + Y_t^2$:
\begin{align}
&Y_{t+1}^1 = Y_t^1\biggl(\frac12 + \frac{X_{t+1}}{Y_t^1 +
(1+\frac1t)Y_t^2}\biggr),\quad Y_{t+1}^2 = Y_t^2\biggl(\frac{t-1}{2t} +
\frac{(1+\frac1t)X_{t+1}}{Y_t^1 +
(1+\frac1t)Y_t^2}\biggr),\notag\\
&W_{t+1} = W_t - \frac{Y_t^2}{2t}.\label{ex-W}
\end{align}
It is easy to see that we have $\zeta_t(W_{t-1}) = W_{t-1}/2$, and hence
$\hat \lambda_t = 1/2$. Thus, the strategy of the first investor is relative
growth optimal.

Let, as always, $r_t^m = Y_t^m/W_t$, $m=1,2$. According to
Theorem~\ref{th-1}, there exists the limit $r_\infty^2 = \lim_{t\to\infty}
r_t^2 \in [0,1)$. We will now show that $r_\infty^2>0$ and $W_\infty =0$.
From~\eqref{ex-W}, we find
\[
r_{t+1}^2 = r_t^2 (1-\alpha_t), \qquad \text{where}\ \alpha_t =
\frac{r_t^2(1-r_t^2)}{2t^2 + tr_t^2 - (r_t^2)^2}.
\]
It is easy to see that $\alpha_t\in(0,1)$ and $\alpha_t = O(t^{-2})$ as $t\to\infty$.
Hence, there exists the limit $r_\infty^2>0$. Also, from \eqref{ex-W} we
have $W_{t+1} = W_t(1- {r_t^2}/(2t))$.
Since $\sum_t {r_t^2}/{t} =\infty$, we have $W_\infty=0$.

However, there is a trivial strategy that guarantees that the wealth does
not vanish: $\lambda_t=0$ for all $t$. \hfill$\square$
\end{example}

Now we turn to analysis of the situation when all the investors use the
relative growth optimal strategy. Obviously, in this case the relative
wealth of the investors will stay constant. Our goal will be to
investigate the asymptotics of the total wealth $W_t$.

To avoid uninteresting complications, let us assume from now on that
$\rho_t(\omega)>0$ for all $t,\omega$. Introduce the discounting sequence
\[
D_t = \rho_1\cdot\ldots\cdot\rho_t, \qquad D_0=1,
\]
and denote by $W_t' = W_t/D_t$ the discounted wealth of the investors, and
by $X_t' = X_t/D_t$ the discounted payoffs.

First, we will show that $W_t'$ does not decrease in the sense that $W_t'$
is a generalized submartingale. But then one can ask the question: will
$W_t'$ asymptotically grow to infinity (provided that $\sum_t |X_t| =
\infty$)? The answer turns out to be quite interesting. We consider it only
in the case when the discounted payoffs $X_t'$ are i.i.d., and show that
$W_t'\to\infty$ if $X_t$ are truly random (i.e. the support of the
distribution of $X_t$ contains more than one point), while $W_t'$ stays
bounded if $X_t'$ are non-random.

\begin{theorem}
\label{th-4}
Suppose all the investors use the strategy $\hat\Lambda$. Then the
following claims are true.

1) The sequence $1/W_t'$ is a
supermartingale, $W_t'$ is a generalized submartingale, and
there exists the limit $ W_\infty' := \lim_{t\to\infty} W_t' \in(0,\infty]$
a.s.

2) Assume additionally that $\rho_t=\rho>0$ for all $t$, where $\rho$ is a
constant; $X_t'$ is a sequence of i.i.d.\ random vectors; and the filtration
$\FF$ is generated by $X_t$, i.e.\ $\F_t= \sigma(X_s,s\le t)$. If $X_t$ are
not equal to a constant vector a.s., then $ W_\infty' = \infty$ a.s.;
otherwise $W_t'=W_0\vee |X_1|/\rho$ for all $t\ge 1$.
\end{theorem}

\begin{proof}
From~\eqref{proof-W} we find that $W_t'$ satisfies the equation
\begin{equation}
\label{growth-W}
W_t' = (1 - |\hat \lambda_t|)W_{t-1}' + |X_t'|
\end{equation}
(here $\hat\lambda_t$ denotes the realization of the relative growth
optimal strategy). 

To simplify the proof, let us first show that it can be reduced to the
case when $\rho_t=1$ for all $t$. Indeed, consider the two markets: the
first one is defined by the sequences $X_t^{(1)}$, $\rho_t^{(1)}$ and
initial capital $Y_0^{(1)}>0$, while the second one by the sequences
$X_t^{(2)} = X_t^{(1)}/D_t^{(1)}$, $\rho_t^{(2)}=1$, and
$Y_0^{(2)}=Y_0^{(1)}$. Assume all the investors in the both markets use the
relative growth optimal strategy. Denote the total wealth in these markets
by $W_t^{(1)}$ and $W_t^{(2)}$, respectively.

It is not hard to see that $\Gamma_t^{(2)} =
\{(\omega,c/D_{t-1}^{(1)}(\omega)) : (\omega,c) \in \Gamma_t^{(1)}\}$, and
$\zeta^{(2)}_t(\omega,c) = \zeta_t^{(1)}(\omega,cD_{t-1}^{(1)}(\omega))
/D_{t-1}^{(1)}(\omega)$. From this and \eqref{optimal}, by induction, we
find that $\hat\lambda_t^{(2)} = \hat\lambda_t^{(1)}$ and $W_t^{(2)} =
W_t^{(1)}/D_t^{(1)}$. Thus, the discounted wealth in the original market
specified by $X_t$ and $\rho_t$ will be the same as the wealth in the market
specified by $X_t'$ and $\rho_t'=1$. So, from now on we may assume
$\rho_t=1$ and $X_t'=X_t$, $W_t'=W_t$ for all $t$.

From \eqref{proof-K}, it follows that $W_t>0$ for all $t$.
Let $V_t = 1/W_t$. From \eqref{growth-W}, we  find
\[
\frac{V_t}{V_{t-1}} =\frac{1}{1-|\hat \lambda_t| + V_{t-1}|X_t|} =
\frac{W_{t-1}}{\zeta_t(W_{t-1}) + |X_t|},
\]
where we used that $1-|\hat \lambda_t| = V_{t-1} \zeta_t(W_{t-1})$. By the
construction of $\zeta_t$, we have $\E(V_t/V_{t-1}\mid \F_{t-1}) \le 1$,
hence $V_t$ is a generalized supermartingale, and, hence, a usual
supermartingale since it is non-negative. This also implies that $W_t$ is a
generalized submartingale (via the identity $W_t = \exp(-\ln V_t)$ and
Jensen's inequality). Moreover, since a non-negative supermartingale has a
finite limit, there exists $V_\infty = \lim_t V_t \in[0,\infty)$ a.s., and
consequently there exists $W_\infty \in (0,\infty]$ a.s. This finishes the
proof of the first claim of the theorem.

To prove the second claim, we will need the following auxiliary result on
convergence of positive supermartingales, which is a corollary from
Proposition~7.1 in~\cite{KaratzasKardaras07}. Suppose $S_t$ is a strictly
positive scalar supermartingale, and consider the generalized
supermartingale $Z_t$ defined by
\[
\Delta Z_t := Z_t - Z_{t-1} =  \frac{S_t}{S_{t-1}} - 1, \qquad Z_0=0.
\]
Denote by $A_t$ the compensator of $Z_t$. Let $h(x) = x^2\wedge |x|$ and
introduce the predictable sequence $H_t$ by
\[
\Delta H_t = \E(h(\Delta Z_t) \mid \F_{t-1}), \qquad H_0=0.
\]
Then we have
\[
\{\omega : \lim_{t\to\infty} S_t(\omega) = 0\} = \{\omega:\lim_{t\to\infty}
(A_t(\omega)+ H_t(\omega)) = \infty\} \text{ a.s.}
\]

We will apply this result to $S_t= V_t$, so that
\[
\Delta Z_t = \frac{W_{t-1}}{\zeta_t(W_{t-1}) + |X_t|} - 1.
\]
Observe that since $X_t$ are i.i.d.\ random vectors, the function
$\zeta_t(\omega,c)$ and the conditional distribution $K_t(\omega,dx)$ can be
chosen not depending on $\omega,t$, hence we will write them simply as
$\zeta(c)$ and $K(dx)$. Then $\Delta H_t = g(W_{t-1})$ with the function
\[
g(c) = \int_{\RNP} h\biggl(\frac{c}{\zeta(c) + |x|} - 1\biggr) K(dx), \qquad
c>0.
\]
It is not hard to check that $g(c)$ is continuous on $(0,\infty)$. Moreover,
if $X_t$ are non-constant (so the support of $K(dx)$ contains more than one
point), then $g(c)>0$ for all $c>0$. Therefore, $g(c)$ is separated from
zero on any compact set $B\not\ni \{0\}$. So, on the set $\{\omega : \lim_t
W_t(\omega) < \infty\}$ we have $H_\infty = \infty$ a.s., which implies
$\P(\lim_t W_t < \infty) = 0$ in the case when $X_t$ are non-constant.

When $X_t$ are constant, $X_t\equiv X \in \R_+^N$, it is easy to see from
the definition of $\zeta$ that
\[
(1-|\hat\lambda_t|) W_{t-1} = \zeta(W_{t-1}) = (W_{t-1} - |X|)^+.
\]
Then \eqref{growth-W} implies $W_t=W_{t-1}\vee |X|$, and, hence, $W_t =
W_0\vee |X|$. 
\end{proof}

\begin{example}
The proved result leads to an observation, which at first seems
counter-intuitive: if $X_t$ and $\tilde X_t$ are two sequences of payoffs
such that $X_t \ge \tilde X_t$ for all $t$, it may happen that the wealth of
investors will grow faster under the smaller sequence $\tilde X_t$.

As an example, let $\rho_t\equiv 1$ and $X_t \equiv X\in
\R_+^N\setminus\{0\}$ be the same constant vector, while $\tilde X_t =
X\xi_t$, where $\xi_t\in[0,1]$ are i.i.d.\ non-constant random variables.
Then, if all the investors use the relative growth optimal strategy, under
the sequence $X_t$ the wealth becomes $W_0\wedge |X|$ after $t=1$ and stop
growing, but it grows to infinity under $\tilde X_t$.

This can be explained by that the presence of randomness in $\tilde X_t$
prevents the investors from ``betting too much''.
\end{example}

\section{Conclusion}
We studied a model of a market where several investors compete for payoffs
yielded by short-lived assets. The main result of the paper consists in
proving that there exists (and is unique) an investment strategy -- the
relative growth optimal strategy -- such that the sequence of its relative
wealth is a submartingale for any strategies of competing investors.

This strategy has a number of other optimality properties. It forms a
symmetric Nash equilibrium when all the investors maximize their expected
relative wealth. It is also a survival strategy in the sense that its
relative wealth always stays separated from zero with probability one on the
whole infinite time interval. Moreover, its relative wealth tends to 1 if
the representative strategy of the other investors is asymptotically
different from it. It is also shown that the relative growth optimal
strategy posses properties similar to growth optimal strategies
(num\'eraires) in single-investor market models, in particular, it maximizes
the asymptotic and $t$-step growth rate of wealth.

Our paper extends the model of \cite{AmirEvstigneev+13} to a market with a
bank account (or a risk-free asset). Inclusion of a bank account in the
model leads to interesting analysis of the asymptotics of the absolute
wealth of investors. In particular, it turns out that the relative growth
optimality (or survival) property of a strategy does not necessarily imply
that its absolute wealth will grow if the competitors use ``bad''
strategies.

\phantomsection
\addcontentsline{toc}{section}{References}%
\bibliographystyle{apalike}
\bibliography{relative_growth_optimal}

\end{document}